\documentclass[11pt]{article}
\usepackage{geometry}                
\usepackage{graphicx}
\usepackage{fullpage, amsmath,amsthm, amssymb, amsfonts}
\usepackage{amsthm,amsmath}
\newtheorem{theorem}{Theorem}

\newtheorem*{wp_conjecture}{Weak Parity Conjecture~\cite{AABB}}

\newtheorem*{cfgs_conjecture}{Hypercube Induced-Degree Conjecture}

\newtheorem{conjecture}[theorem]{Conjecture}

\newcommand{\bE}{\mathbb{E}}

\title{A Note on a Communication Game}
\author{Andrew Drucker\thanks{University of Chicago, Computer Science Department. Email: andy.drucker@gmail.com.  This research was conducted at the Institute for Advanced Study and was supported by the NSF under agreements Princeton University Prime Award No. CCF-0832797 and Sub-contract No. 00001583.}} 
\date{}                                           

\begin{document}
\maketitle

\begin{abstract}
We describe a communication game, and a conjecture about this game, whose proof would imply the well-known Sensitivity Conjecture asserting a polynomial relation between sensitivity and block sensitivity for Boolean functions.  The author defined this game and observed the connection in Dec. 2013 - Jan. 2014.  The game and connection were independently discovered by Gilmer, Kouck\'{y}, and Saks, who also established further results about the game (not proved by us) and published their results in ITCS `15~\cite{GKS_itcs}.  

This note records our independent work, including some observations that did not appear in~\cite{GKS_itcs}. Namely, the main conjecture about this communication game would imply not only the Sensitivity Conjecture, but also a stronger hypothesis raised by Chung, F\"{u}redi, Graham, and Seymour~\cite{CFGS}; and, another related conjecture we pose about a ``query-bounded'' variant of our communication game would suffice to answer a question of Aaronson, Ambainis, Balovis, and Bavarian~\cite{AABB} about the query complexity of the ``Weak Parity'' problem---a question whose resolution was previously shown by~\cite{AABB} to follow from a proof of the Chung et al. hypothesis.
\end{abstract}

\section{The Carmen Sandiego game}

The work described here\footnote{The ideas of the present paper were formulated in Dec. 2013 - Jan. 2014, and discussed (in varying detail) in early 2014 with Scott Aaronson, Laci Babai, Sasha Razborov, and Avi Wigderson.} concerns the following ``game'' involving three parties: Alice, Bob, and Carmen.

Carmen Sandiego is a notorious globe-trotting criminal (created by Br{\o}derbund Software, and now owned by The Learning Company) whose misdeeds inadvertently raise awareness of geography and world cultures.  In her latest caper, Carmen makes a single visit to each of $n > 1$ world cities $c_1, \ldots, c_n$, in some order of her choosing.  Alice and Bob are two sleuths cooperating to catch Carmen; they know the identities of the $n$ cities but not the order in which they will be visited.  

Bob is hot on the trail and follows Carmen directly from city to city but, due to severe funding restrictions, can't communicate directly with Alice as he travels.  Instead, he is only able to leave a single 0/1 ``clue-bit'' behind in each city, based on what he has observed of Carmen's travels so far.  (He chooses his clue-bits based upon some deterministic algorithm or ``strategy'' STRAT$_B$ agreed upon with Alice in advance.)  To make matters worse, when Bob follows Carmen to the final city visited on her tour (let it be $c_{j}$), she captures him and plants the clue-bit for this city herself, with an intention to confuse.  Carmen then ``hides out'' with her captive somewhere in $c_j$.

Afterwards Alice, having witnessed none of these events, visits all $n$ cities and observes the clue-bit $b_i$ left in each city $c_i$.  Alice then prepares a list of ``suspect'' cities $S \subseteq [n]$ consisting of all cities that could potentially be Carmen's final hiding place, based on Alice's view of $b = (b_1, \ldots, b_n)$ and her knowledge of the algorithm followed by Bob.

The ``cost'' of this outcome is defined as the set size $|S|$---indicating the number of cities that might have to be thoroughly searched to uncover and arrest Carmen.  The ``complexity'' of STRAT$_B$ is defined as the maximum cost over all possible outcomes, ranging over possible actions by Carmen.

The above game was independently defined and studied by Gilmer, Kouck\'{y}, and Saks~\cite{GKS_itcs}, who also established further results about the game (not proved by us) and published their results in ITCS `15.  (The authors of~\cite{GKS_itcs} describe this game with inessential narrative differences---no ``Carmen'' is present, and their Alice and Bob play different roles.  We choose to follow our original description.)

This author conjectured that any STRAT$_B$ used by Bob has complexity at least $\alpha \cdot n^\beta$, for some absolute constants $\alpha, \beta > 0$.  This still seems plausible, although as we will discuss,~\cite{GKS_itcs} have disproved some stronger versions of this statement; their results also call the original statement's likelihood into question.  We also note that Szegedy~\cite{Sze15} has shown a Bob-strategy with complexity $O(n^{0.4732})$, improving on a $.8\sqrt{n}$ upper bound achieved in~\cite{GKS_itcs}.

\paragraph{A more formal setup:} We will give definitions and then explain them.  Fix $n > 1$ as before, and let $S_n$ be the set of all permutations $\pi = (\pi(1), \ldots, \pi(n))$ over $[n]$.

  A \emph{Bob-strategy} STRAT$_B$ is defined as a family of functions $F_t: S_n \rightarrow \{0, 1\}$ for $t = 1, 2, \ldots, n-1$.  We always require that each $F_t$ is \emph{$t$-restricted}, meaning that $F_t(\pi)$ is a function of the first $t$ values $\pi(1), \ldots, \pi(t)$.  

With reference to a fixed Bob-strategy STRAT$_B$, for a permutation $\pi \in S_n$ and a $z \in \{0, 1\}$ we define $b(\pi, z) \in \{0, 1\}$ as the string whose $j^{th}$ coordinate is given by
$$
b_j \ := \
\begin{cases}
F_t(\pi), \ \text{if } j = \pi(t) \text{ for some $t < n$,} \\
z,  \ \text{\ \ \ \ \ \ if } j = \pi(n) .\\
\end{cases}
$$
For $b \in \{0, 1\}^n$, we define the set $S = S(b) \subseteq [n]$ of ``suspect cities'' as 
\[ S \ := \ \{i: \exists \rho \in S_n, u \in \{0, 1\} \text{ such that } b(\rho, u) = b \} \ . \] 
The \emph{cost} of strategy STRAT$_B$ on pair $(\pi, z)$ is defined as $| S(b^\star)| $, where $b^\star := b(\pi, z)$.
The \emph{complexity}, $\mathrm{compl}($STRAT$_B)$, is defined as the maximum cost over all pairs $(\pi, z)$ as above.  The complexity of the Carmen Sandiego communication game for $n$ cities, denoted $\mathrm{compl}(\mathcal{G}_{CS, n})$, is defined as
\[   \mathrm{compl}(\mathcal{G}_{CS, n}) \ := \ \min\ \{ \mathrm{compl}(\mathrm{STRAT}_B)\} \ ,  \]
with the minimum ranging over all Bob-strategies STRAT$_B$ for $n$ cities.

The intended ``interpretation'' of these definitions is as follows: 

\begin{itemize}
\item We interpret each permutation $\pi$ as a possible itinerary for Carmen, where $\pi(i)$ indicates the index of the $i^{th}$ city visited.

\item The function $F_t$ tells Bob which clue-bit to leave at the $t^{th}$ city on his travels, the city $c_{\pi(t)}$; this clue must be determined based only on what he has seen of Carmen's itinerary so far---that is, based on the values $\pi(1), \ldots, \pi(t)$.  As Bob is captured before he can choose a final clue, his $n^{th}$-step rule becomes irrelevant and is omitted.

\item The bit $b_j$ is the clue that Alice finds in the $j^{th}$ city $c_j$, if Carmen followed the itinerary $\pi$ and if she chose $z$ as the ``confusing'' clue to put in her final hideout-city $c_{\pi(n)}$.

\end{itemize}

\section{Relation to a question on the hypercube}  

Next we recall the basic notion of induced subgraphs.  If $G = (V, E)$ is an undirected graph and $K \subseteq V$, the \emph{induced subgraph} $G[K]$ is defined as the graph that has vertex set $K$ and edge set $E' := \{(u, v) \in E: u , v \in K\}$.  The next conjecture considers the case where $G = H_n = \{0, 1\}^n$ is the Boolean hypercube, whose vertices $x, x'$ are adjacent iff $x, x'$ differ in exactly one coordinate. 

Chung, F\"{u}redi, Graham, and Seymour~\cite{CFGS} raised the question of whether the following conjecture holds:

\begin{cfgs_conjecture}\label{cfgs-conj} There are absolute constants $a, b > 0$ such that the following holds for all $n > 0$: if $K \subseteq H_n$ is any set of size $|K| > 2^{n-1}$, then there is a $y \in K$ whose degree within $H_n[K]$ is at least $a \cdot n^b$.
\end{cfgs_conjecture}

The Hypercube Induced-Degree Conjecture was shown by Gotsman and Linial~\cite{GL92} to imply the well-known Sensitivity Conjecture, which asserts that two measures of Boolean function complexity, the \emph{sensitivity} $s(f)$ and \emph{block sensitivity} $bs(f)$ (the latter defined by Nisan in~\cite{Nis91}), are always within a polynomial factor of each other.  The question of the Sensitivity Conjecture was raised in the conference version of Nisan's paper~\cite{Nis89}.   For background and variations on the Sensitivity Conjecture (including several equivalent conjectures) we recommend the survey of Hatami, Kulkarni, and Pankratov~\cite{HKP11}.

We prove:

\begin{theorem}\label{thm:cs_to_cfgs} Suppose that $K \subseteq H_n$ is of size $|K| > 2^{n-1}$ and yet every $y \in K$ has degree at most $D > 0$ within $H_n[K]$.  Then, $\mathrm{compl}(\mathcal{G}_{CS, n}) \leq D$.
\end{theorem}

As an immediate consequence, we see that if $\mathrm{compl}(\mathcal{G}_{CS, n})\geq n^{\Omega(1)}$ then the Hypercube Induced-Degree Conjecture is true.  The authors of~\cite{GKS_itcs} showed (independently) that $\mathrm{compl}(\mathcal{G}_{CS, n})\geq n^{\Omega(1)}$ implies the Sensitivity Conjecture.  Their proof technique is different and involves analyzing representations of Boolean functions as real polynomials.

\begin{proof}[Proof of Theorem~\ref{thm:cs_to_cfgs}] We will define a Bob-strategy STRAT$_B = (F_1, \ldots, F_{n-1})$ based on the fixed set $K$ from our starting assumption.  First, for any $w \in \{0, 1, *\}^n$, we will use
\[ K(w) \ := \ \{y \in K: \text{for every $i$ such that $w_i \in \{0, 1\}$, we have $y_i = w_i$}\}  \]
to denote the set of strings in $K$ ``agreeing with'' $w$.
 For $t \geq 0$ we will let $w^t = w^t(\pi) \in \{0, 1, *\}^n$ denote the vector of clue-bits left by Bob in $c_1, \ldots, c_n$ after $t$ steps of Carmen's itinerary $\pi$; here, we take $w^t_\ell = *$ if Carmen has not visited $c_i$ in the first $t$ steps, that is, if $\pi^{-1}(\ell) > t$.  Thus $w^0 = *^n$.

For $w = (w_1, \ldots, w_n) \in \{0, 1, *\}^n$ and $u \in \{0, 1\}$ we define $w[i \leftarrow u]$ as $w$ with the $i^{th}$ coordinate set to $u$.  We define $F_t$ inductively for $t \geq 1$ by taking
\[   F_t(\pi) \ = \ F_t(\pi(1), \ldots, \pi(t)) \ := \ \arg \max_{u \in \{0, 1\}} \left|  K\left(w^{t - 1}[\pi(t) \leftarrow u]\right)  \right|  \ .  \]
That is, we set the $t^{th}$ clue-bit (placed at position $\pi(t)$) in such a way as to maximize the number of strings in $K$ agreeing with $w^t = w^{t - 1}[\pi(t) \leftarrow u]$, subject to the inductive setting of the previously chosen clue-bits (which determine $w^{t -1}$).  Any ties above are broken by an arbitrary fixed rule---in favor of $u = 0$, for concreteness.

This Bob-strategy clearly has the ``$t$-restricted'' property we require for each $t \leq n-1$.  Let us analyze its behavior.  Fixing a Carmen itinerary $\pi \in S_n$ and a final bit $z \in \{0, 1\}$ determines the strings $w^0, \ldots, w^n$.  We have $|K(w^0)| = |K| > 2^{n-1}$ and we claim that $|K(w^t)| > 2^{n - 1 - t}$ for each $t \in [n-1]$.  This follows easily by induction on $t$, since by our choice of $t^{th}$ clue-bit we have $|K(w^{t})| \geq  .5 |K(w^{t-1})|  $.

Thus, $|K(w^{n-1})| > 2^0 = 1$.  On the other hand, $K(w^{n-1})$ is contained in a subcube of dimension 1, i.e., the set of Boolean strings agreeing with $w^{n-1}$.  Thus $|K(w^{n-1})| = 2$.  It follows that no matter Carmen's choice of final ``confusing clue'' $z$, the singleton $K(w^n) = w^n = b = b(\pi, z)$ must lie in $K$.  

Letting $b^{\oplus i}$ denote $b$ with its $i^{th}$ bit flipped, we conclude that the final index $\ell = \pi(n)$ must be such that $b$ and $b^{\oplus \ell}$ both lie in $K$.  It follows that each city index $\ell' $ in the ``suspect'' set $S = S(b)$ defines a distinct neighbor of $b$ within the induced subgraph of $H_n$ on $K$.  But by our assumption on $K$, there are at most $D$ such neighbors.  Thus we must always have $|S| \leq D$.  It follows that $\mathrm{compl}(\mathcal{G}_{CS, n}) \leq  \mathrm{compl}(\mathrm{STRAT}_B) \leq D$, as claimed.
\end{proof}

We note that in the strategy STRAT$_B$ described above, Bob's choice of clue-bit at step $t$ is fully determined by the string $w^{t - 1} \in \{0, 1, *\}^n$ describing the locations and values of the clue-bits he has already fixed.  Thus to prove the Hypercube Induced-Degree Conjecture, it suffices to lower-bound the complexity of Bob-strategies having this restricted form.  The same observation (toward proving the Sensitivity Conjecture) was made independently in~\cite{GKS_itcs}.

\section{An average-case version of the problem}

Now suppose we assume that Carmen's itinerary $\pi \in S_n$ is a uniformly distributed random variable, so that in particular, the value $\pi(n)$ is uniform over $[n]$.  Also assume that her final ``confusing clue'' bit $z$ is uniform and independent of $\pi$.  In this setting, it is tempting to expect that no matter which strategy STRAT$_B$ Bob uses for his clues, Alice will have significant expected uncertainty about $\pi(n)$, even after seeing the clue-string $b = b(\pi, z)$.

We conjectured that under this experiment, for any fixed STRAT$_B$ the conditional uncertainty of $\pi(n)$ satisfies $H( \pi(n) |  b  )  \geq c \log n$,
for some absolute constant $c > 0$, where $H(\cdot|\cdot)$ denotes the conditional Shannon entropy.  This conjecture would imply and strengthen the conjecture that $\mathrm{compl}(\mathcal{G}_{CS, n}) \geq n^{\Omega(1)}$, and until learning of the work of Gilmer, Kouck\'{y}, and Saks~\cite{GKS_itcs}, this author considered it likely. However, these researchers independently considered this entropic version of the conjecture, and managed to refute it!  They also showed that the worst-case complexity $\mathrm{compl}(\mathcal{G}_{CS, n})$ is sub-polynomial in $n$ in a model where Bob is allowed to leave clues from even the slightly larger alphabet $\Sigma = \{0, 1, 2\}$.  Their results indicate the delicate nature of the question, and cast some doubt on whether $\mathrm{compl}(\mathcal{G}_{CS, n}) \geq n^{\Omega(1)}$ in the Boolean model.  These authors do, however, raise an ``average-cost'' strengthening of this hypothesis that remains open (see~\cite{GKS_itcs}).

\section{The Carmen Sandiego search problem with query bounds}  
For the Carmen Sandiego game defined as before, we next consider a different model for Alice's behavior.   In this model, Alice doesn't have enough time or money to visit every city.  We model this formally by requiring that Alice has only \emph{query-bounded} access to the clue string $b = b(\pi, z) \in \{0, 1\}^n$ defined by $\pi, z$, and by the Bob-strategy STRAT$_B$.  Recall that for $0 < t \leq n$, a \emph{randomized $t$-query algorithm} for Alice (which we will denote as $R_A$) is a probability distribution over depth-$t$ decision trees $\{R_{A, r}\}_{r \in \Omega}$ on the $n$ input variables $b_1, \ldots, b_n$.  

In the query-bounded setting we also focus on a modified success criterion for Alice, in which she is no longer trying to give a short list of cities where Carmen might be hiding; now her only goal is to \emph{make a query} to the bit $b_{\pi(n)} = z$ planted by Carmen in the final city $c_{\pi(n)}$.  (In this variant of the model, we imagine that this is enough for Alice to detect and thwart Carmen.)  Formally, let VISITS$_A \subseteq [n]$ denote the random variable giving the $t$-subset of coordinates queried by Alice; we say that an execution of $R_A$ on $b = b(\pi, z)$ is \emph{search-successful} if $\pi(n) \in $ VISITS$_A$.  Note that here, there is no requirement for Alice to \emph{identify} which of her queries goes to the $\pi(n)^{th}$ coordinate (or even to decide if success has occurred).

We conjecture:

\begin{conjecture}\label{conj:catch_hard} There are $\alpha, \beta  > 0$ such that the following holds for any Bob-strategy STRAT$_B$ and any randomized algorithm $R_A$ for Alice making $t  \leq  \alpha \cdot n^\beta $ queries.
If $(\pi, z)$ are drawn uniformly from $S_n \times \{0, 1\}$, and $\mathbf{b} = b(\pi, z)$, then 
\[ \Pr_{\pi, z, R_A}[R_A \text{ is search-successful on }\mathbf{b}] \ < \ 1/3 \ . \]
\end{conjecture}

By an argument similar to that of Theorem~\ref{thm:cs_to_cfgs}, we show that this conjecture's truth would imply another conjecture of Aaronson, Ambainis, Balodis, and Bavarian~\cite{AABB}, about the query complexity of computing the Parity function with high confidence on an arbitrary set of strictly more than half of all possible inputs.  

\begin{wp_conjecture} There are absolute constants $\alpha, \beta > 0$ such that the following holds for all $n > 0$.  Suppose $R$ is a randomized algorithm on $n$-bit Boolean input strings making $t \leq \alpha \cdot n^\beta$ queries to its inputs.   Let $K \subseteq \{0, 1\}^n$ be the set of inputs on which $A$ computes the Parity function PAR $=$ PAR$_n$ with at least $2/3$ success probability,
\[  K \ := \ \left\{ y:  \ \Pr[R(y) = PAR(y)] \ \geq \ 2/3 \right\}  \ . \]
Then we must have $|K| \leq 2^{n-1}$.
\end{wp_conjecture}

The authors of~\cite{AABB} show that the Hypercube Induced-Degree Conjecture implies the Weak Parity Conjecture.  Our definition and study of the Carmen Sandiego game was initially conceived in an attempt to prove their Weak Parity Conjecture, through the connection given in the next result.

\begin{theorem}\label{thm:connect} Let $0 < t \leq n$.  If $R$ is a $t$-query randomized algorithm on $n$ input bits, and $K \subseteq \{0, 1\}^n$ is a set of strictly more than $2^{n-1}$ inputs $y$ on which $\Pr[R(y) = PAR(y)] \geq 2/3$, then there is a choice of Bob-strategy STRAT$_B$ such that, for the Alice query strategy $R_A := R$, the following holds. If $(\pi, z)$ are drawn uniformly from $S_n \times \{0, 1\}$, and $\mathbf{b} = b(\pi, z)$, then we have
\[  \Pr[R_A(\mathbf{b})\text{ is search-successful}] \ \geq \ 1/3  \ .     \]
\end{theorem}

As a consequence of Theorem~\ref{thm:connect}, we see that Conjecture~\ref{conj:catch_hard} implies the Weak Parity Conjecture.  This work also has a simple takeaway message which can be studied even without talk of Alice, Bob, and Carmen: for any set $K \subseteq \{0, 1\}^n$ of more than half of all strings, we propose the random variable $\mathbf{b} = b(\pi, z)$ used in the proof of Theorem~\ref{thm:connect} as a candidate hard distribution (supported entirely on $K$) for query-bounded algorithms attempting to compute the Parity function.  This seems to us a promising approach to the Weak Parity Conjecture.

\begin{proof}[Proof of Theorem~\ref{thm:connect}]  We use the same Bob-strategy STRAT$_B$ (defined relative to the set $K$) as in the proof of Theorem~\ref{thm:cs_to_cfgs}.  Letting the random variable $\mathbf{b} = b(\pi, z)$ be as above, define the modified string $\mathbf{b}' := b(\pi, \overline{z})$ with flipped final clue-bit $\overline{z}$.  Now $\mathbf{b}$ and $\mathbf{b}'$ are identically distributed (since $z$ is uniform and independent of $\pi$), so considering the Boolean output of the query algorithm $R$, we have 
\begin{equation}\label{eq:11} \Pr[R(\mathbf{b})= PAR(\mathbf{b})]  \ = \  \Pr[R(\mathbf{b}')= PAR(\mathbf{b}')]  \ . \end{equation}
Also, for a fixed choice of randomness $r$ of $R$ (determining a $t$-query decision tree $R_r$ applied to the input), the execution of Alice's query algorithm $R_{r}(\mathbf{b})$ is search-successful if and only if $R_{r}(\mathbf{b}') $ is search-successful.

On the other hand, if $R_{r}(\mathbf{b})$ is not search-successful, then its view of $\mathbf{b}$ is identical to that of $\mathbf{b}'$ in the execution $R_r(\mathbf{b}')$, since $\mathbf{b}, \mathbf{b}'$ are identical outside of the $\pi(n)^{th}$ position.  Thus the Boolean outputs of these two computations are identical.  It also clearly holds that $PAR(\mathbf{b}) \neq PAR(\mathbf{b}')$, so at most one of the computations $R_r(\mathbf{b}), R_r(\mathbf{b}')$ can correctly output the Parity function of its input unless $R_{r}(\mathbf{b})$ is search-successful.  Thus in terms of indicator random variables, we have
\begin{equation}\label{eq:22}    \mathbf{1}[R(\mathbf{b})= PAR(\mathbf{b})] \ + \   \mathbf{1}[R(\mathbf{b}')= PAR(\mathbf{b}')]  \ \leq \ 1 +  \mathbf{1}[ R(\mathbf{b})\text{ is search-successful}] \ .   \end{equation}
Using Eqs.~(\ref{eq:11}) and~(\ref{eq:22}), we find that 
\begin{align*}   \Pr[R(\mathbf{b})= PAR(\mathbf{b})] \ &= \ .5\left( \ \bE\left[  \mathbf{1}[R(\mathbf{b})= PAR(\mathbf{b})] \ + \   \mathbf{1}[R(\mathbf{b}')= PAR(\mathbf{b}')]  \right] \  \right) \\ &\leq \ .5(1 + \Pr[R(\mathbf{b})\text{ is search-successful}]) \ .   \end{align*}
Now $[\mathbf{b} \in K]$ always holds, by the design of STRAT$_B$ and our analysis from the proof of Theorem~\ref{thm:cs_to_cfgs} (using the fact that $|K| > 2^{n-1}$).  By our other initial assumption on $K$ in the present context, we then have $ \Pr[R(\mathbf{b})= PAR(\mathbf{b})] \geq 2/3$.  Combining and rearranging shows that $\Pr[R_A(\mathbf{b})\text{ is search-successful}]) \geq 1/3$ as needed, proving Theorem~\ref{thm:connect}.
\end{proof}

\section*{Acknowledgements}
I thank Laci Babai for encouragement to produce this note.  I thank Scott Aaronson for sharing an early draft of~\cite{AABB} and for many interesting conversations about the Sensitivity Conjecture and related topics.

\end{document}